\documentclass[letter paper,conference,10 pt]{IEEEtran}
\IEEEoverridecommandlockouts
\usepackage{amsmath,amssymb,amsfonts}
\usepackage{amsthm}
\usepackage{algorithmic}
\usepackage{graphicx}
\usepackage{textcomp}
\usepackage{xcolor}

\newtheorem{Theorem}{Theorem}

\newtheorem{Lemma}{Lemma}
\newtheorem{Problem}{Problem}

\newtheorem{Remark}{Remark}
\newtheorem{Assumption}{Assumption}
\def\BibTeX{{\rm B\kern-.05em{\sc i\kern-.025em b}\kern-.08em
    T\kern-.1667em\lower.7ex\hbox{E}\kern-.125emX}}
    
\newcommand{\m}[1]{\boldsymbol{#1}}
\newcommand{\mc}[1]{\mathcal{#1}}
\newcommand{\mb}[1]{\mathbb{#1}}
\newcommand{\mcl}[1]{\mathcal{\m{#1}}}

\begin{document}

\title{Decentralized sliding-mode control laws for the bearing-based formation tracking problem
\thanks{$^*$Corresponding author: \texttt{minh.trinhhoang@hust.edu.vn}}
}

\author{\IEEEauthorblockN{Dung Van Vu}
\IEEEauthorblockA{\textit{Flight Instrument Center} \\
\textit{Viettel High Technology Industries Corporation}\\
Hanoi, Vietnam \\
0000-0002-1351-6191}
\and
\IEEEauthorblockN{Minh Hoang Trinh$^*$}
\IEEEauthorblockA{\textit{School of Electrical Engineering} \\
\textit{Hanoi University of Science and Technology (HUST)}\\
Hanoi, Vietnam \\
0000-0001-5736-6693}
}

\maketitle

\begin{abstract}
This paper studies the time-varying bearing-based tracking of leader-follower formations. The desired constraints between agents are specified by bearing vectors, and several leaders are moving with a bounded reference velocity. Each followers can measure the relative positions of its neighbors, its own velocities, and receive information from their neighbors. Under the assumptions that the desired formation is infinitesimally bearing rigid and the local reference frames of followers are aligned with each other, two control laws are presented in this paper based on sliding mode control approach. Stability analyses are given based on Lyapunov stability theory and supported by numerical simulations. 
\end{abstract}


\section{Introduction}
In recent years, there is much research on multi-agent formation control due to its applications in wide-range fields such as swarms of unmanned aerial vehicles (UAV), robot formations, and networks of satellites \cite{anderson2008rigid,oh2015survey}. By cooperative decision-making, distributed control of each agent in a formation can be performed without a central brain/computer, which reduces the production costs. For example, a team of low-cost UAVs can survey and monitor a wider area than a higher-cost UAV does.

According to the measured, the controlled variables and the desired constraints between agents \cite{oh2015survey}, the formation control problem is subdivided into four general categories, including position-, displacement-\cite{Cao2011TAC}, distance-\cite{Cao2011,Chen2017global,Vu2020ICERA,Vu2020ICCAS}, and bearing-based control \cite{Bishop2011a,Eren2012,zhao2019bearingMag}. In the literature, many papers considered the displacement-based formation control, which assumes that agents can measure the relative positions of their neighbors and try to satisfy the desired displacement constraints  \cite{oh2015survey,wang2019integral,yu2016second,pilloni2013finite}. Various control approaches have been proposed to solve the displacement-based formation, such as the consensus law, PID control law, and sliding mode control \cite{wang2019integral,yu2016second,pilloni2013finite}. In the displacement-based formation, it is required that the agents being equipped with expensive sensors or a common reference frame so that they can sense their relative positions. Hence, the distance-based formation has been studied, which relaxes the requirement of a common reference frame. However, it is well-known that undesired equilibria exist in distance-based formation control \cite{Sun2014IFAC}, and (almost) global stabilization can only be shown for very specific formations, such as three-agent \cite{Cao2011} or Laman-generated formations \cite{Chen2017global} in 2-dimensional space or tetrahedron formation in 3-dimensional space \cite{Park2014stability}. With additional sensing capability such as relative attitude, the distance-based formation control can be converted to displacement-based, for examples, see \cite{Oh2013OA,Montijano2016}. Another approach is bearing-based formation, in which the desired constraints are given as a set of bearing vectors \cite{Eric2014,Zhao2015,trinh2018bearing,Tron2016distributed}. The bearing vector contains directional information between agents and can be obtained from cameras attached on the body of each agents. In a bearing-only formation control problem, the agents can only sense the bearing vectors with regard with its neighboring agents. In this paper, by bearing-based formation control, we assume that the agents can sense their displacements and control their positions so that a set of desired bearing vectors are satisfied. Several works were conducted in the literature to tackle this problem based on the notion of bearing rigidity theory \cite{Zhao2015,li2018bearing,li2020adaptive}. The authors in \cite{tang2020bearing} proposed a bearing-persistent excitation condition for reducing number of measurements between agents such that a desired formation is still achieved. The bearing-based maneuver problems, where the agents must achieve the desired bearings and move in a same velocity have been recently studied \cite{zhao2020bearing,zhao2019bearing,trinh2021robust,li2019bearing}. However, the above-mentioned works assume that the leaders' velocity is constant, which reduces its applicability  \cite{zhao2020bearing,zhao2019bearing,trinh2021robust,li2019bearing}. It is also worth noting that bearing-only formation tracking with time-varying leaders' velocity has been studied in \cite{Trinh2021LCSS}. However, \cite{Trinh2021LCSS} is restricted to a directed leader-follower formation \cite{Trinh2021LCSS} and single-integrator modeled agents.  

In this paper, we study the bearing-based formation tracking problem for leader-follower formations. The desired formation is specified by a set of bearing vectors and satisfies the infinitesimally bearing rigid condition \cite{Zhao2015}. There exist at least two leaders moving with a common bounded time-varying velocity. It is assumed that leaders can measure their positions and act as moving beacons during the maneuver. The remaining agents, called followers, modeled by double-integrator and have pre-aligned local reference frames. The follower can measure their own velocities and the displacements with a set of neighboring agents. Two sliding mode formation tracking control laws are introduced to impose followers to move with the unknown common velocity and hold the desired bearing vectors. The first control law directly uses the displacements to impose the followers satisfy the bearing constraints directly. The second control law assumes that the agents can communicate with their neighbors. Each agent uses  displacements, relative velocities and communicated information to estimate the true position, the desired position based on  sliding mode estimator. A PD-like control law is used for agent to track the estimated position in the formation. Although this approach requires more information, it can be modified to tackle practical issues such as input saturation, control failure, etc. 

The remainder is organized as follows. Background and problem formulation are presented in Section 2.  Section 3 introduces the main results including the proposed control laws and theoretical analysis.  Simulation results are shown in Section 4 and Section 5 concludes the paper.

Notations. Let $\m{x} = [x_1, \ldots, x_d]^\top$ denote a column vector in $\mb{R}^d$. We denote $|\m{x}| = [|x_1|, \ldots, |x_d|]^T$. Let $\|\cdot\|$ denote the 2-norm or Euclidean norm, $\|\m{x}\| = \sqrt{\sum_{i=1}^n x_i^2}$. Further, we use $\|\m{x}\|_1 = \sum_{i=1}^d |x_i|$ and $\|\m{x}\|_{\infty}$ to denote the 1-norm and $\infty$-norm of $\m{x}$, respectively. For $y \in \mb{R}$, the signum function is denoted by $\text{sign}(y)$. When applying to a vector, this notation is understood as an element-wise operator, i.e., $\text{sign}(\m{x}) = [\text{sign}(x_1), \ldots, \text{sign}(x_d)]^\top$. 
\section{Problem formulation}
Consider an $n$-agent leader-follower formation in $\mathbb{R}^d$, $d\ge 2$, where there exist $l\ge 2$ leaders $i\in \mc{L}=\{1,\ldots,l\}$ and $f$ followers $i\in \mc{F}=\{l+1,\ldots,n\}$. In a common global reference frame $^g\Sigma$, we use $\m{p}_i$ and $\m{v}_i \in \mathbb{R}^d$ to denote the position and the velocity of agent $i$. Let each agent $i$ maintain a local reference frame $^i\Sigma$, and it is assumed that these frames are all aligned. Thus, the relative quantities such as the displacements $\m{p}_i^i - \m{p}_j^i$, relative velocities $\m{v}_i^i -\m{v}_j^i$ are the same when measured from any local frame of each agent in the system. Without loss of generality, we assume that these local frames are aligned with $^g\Sigma$, and thus from now on, all equations will be expressed in $^g\Sigma$. Let leaders move with a common velocity $\m{v}_c$, the equation for each leader is given by 
\begin{align}\label{Motion leader}
    \dot{\m{p}}_i= \m{v}_c,~ i =1, \ldots, l.
\end{align}
The followers' models are given by 
\begin{align}\label{Motion follower}
    \dot{\m{p}}_i=\m{v}_i,~ \dot{\m{v}}_i=\m{u}_i, \; i= l+1, \ldots, n,
\end{align}
where $\m{u}_i$ denotes the control input of agent $i$ that will be later designed. 
We use a directed graph $\mc{G}=(\mc{V},\mc{E})$ to characterize the interaction among the agents in the formation, where the agents’ set is $\mc{V}=\mc{L}\cup \mc{F}$ and the edges’ set is $\mc{E}\subset \mc{V}^2$. An edge $(i,j) \in \mc{E}$ implies that $i$ can sense information from agent $j$. The neighbor set $\mc{N}_i$ consists of agents $j$ such that agent $i$ can measure the relative position to $j$ and receive information from $j$. In this paper, we consider a graph $\mc{G}$ such that communication between followers is bidirectional while the communication is unidirectional from follower to leader. The Laplacian matrix $\m{L} = [l_{ij}] \in \mb{R}^{n\times n}$ corresponding to the graph $\mc{G}$ is defined as 
\begin{align*}
    l_{ij} = \begin{cases} 
                -1, &\mbox{if } (i,j) \in \mc{E}, \ i\ne j, \\ 
                ~0, & \mbox{if } (i,j) \notin \mc{E}, i \ne j, \\
                -\sum_{j=1}^n l_{ij}, & \mbox{if } i=j.
             \end{cases}
\end{align*}
Thus, we can rewrite the matrix $\m{L}$ as 
\begin{align} \label{eq:Laplacian}
    \m{L} = \begin{bmatrix}
            \m{0}_{l\times l} & \m{0}_{l\times f}\\
            \m{L}_{fl} & \m{L}_{ff}
            \end{bmatrix}.
\end{align}
where $\m{L}_{ff} \in \mb{R}^{f\times f}$ is a symmetric, positive definite matrix \cite{Olfati2007consensuspieee}. 

By bearing-based, we imply that the agents can sense the relative positions $\m{p}_j-\m{p}_i$ and control the bearing vectors $\m{g}_{ij}=\frac{\m{p}_j-\m{p}_i}{\Vert \m{p}_j-\m{p}_i\Vert}$  to desired time-invariant bearing vectors $\m{g}_{ij}^*, \forall j\in \mc{N}_i$. The following assumption is adopted in this paper:
\begin{Assumption}\label{Asumption 1}
There exists a desired realization $\m{p}^*=\text{vec}\left(\m{p}_1,\ldots,\m{p}_n\right)$ such that $\frac{\m{p}_j^*-\m{p}_i^*}{\Vert \m{p}_j^*-\m{p}^*_i\Vert}=\m{g}_{ij}^*$ and $\m{p}_i^*=\m{p}_i,~\forall i\in \mc{L}$. The desired formation $\left(\mc{G},\m{p}^*\right)$ is infinitesimally bearing rigid \cite{zhao2019bearingMag}. 
\end{Assumption}
The bearing Laplacian matrix $\mcl{B}=\left[\mcl{B}_{ij}\right]$ is defined as $\mcl{B}_{ii}=\sum_{j\in N_i} \m{P}_{\m{g}_{ij}^*}$, $\mcl{B}_{ij}=-\m{P}_{\m{g}_{ij}^*}$ if $j\in \mc{N}_i$, and $\mcl{B}_{ij}=\m{0}_{d\times d}$ in other cases, where the projection matrix $\m{P}_{\m{g}}=\m{I}_d-\m{g}\m{g}^\top$. Rewriting the bearing matrix $\mcl{B}$ as $\mcl{B} = \left(\begin{matrix}\mcl{B}_{ll}&\mcl{B}_{lf} \\ \mcl{B}_{fl}& \mcl{B}_{ff}\\ \end{matrix}\right)$, the following lemma holds:
\begin{Lemma}\label{Lemma 1}
Suppose that Assumption \ref{Asumption 1} holds and $l\ge 2$. The matrix $\mcl{B}_{ff}\in\mathbb{R}^{df\times df}$ is symmetric, positive-definite. 
\end{Lemma}
Next, we adopt the following assumptions about the sensing and capabilities of the agents.
\begin{Assumption}\label{Asumption 2}
Agents measure their velocities and relative positions to their neighbors. Leaders can measure their positions. Agents can exchange several variables with their neighbors. Follower $i$ knows the desired bearing vector $\m{g}_{ij}^*,~ \forall j\in \mc{N}_i$. The common velocity of $l$ leaders satisfy $\Vert \m{v}_c\Vert\le\delta_1$ and $\Vert \dot{\m{v}}_c\Vert\le\delta_2$, where two positive constants $\delta_1$ and $\delta_2$ are known by followers.
\end{Assumption}
We can now state the main problem that will be studied in this paper:
\begin{Problem}
Consider a formation whose motions are described by \eqref{Motion leader}-\eqref{Motion follower} and suppose that Assumptions \ref{Asumption 1}-\ref{Asumption 2} are held. Design $\m{u}_i$ such that $\m{g}_{ij}\to \m{g}^*_{ij},~\forall j\in \mc{N}_i$.
\end{Problem}
\section{Main results}
In this section, two control laws will be proposed to solve Problem 1. These control laws are decentralized in the sense that the agents can compute their own control input from local measured and communicated variables. Since the agents' local reference frames are required to be fixed and aligned with each other, these control laws are categorized as decentralized controllers instead of distributed ones. Note that one may adopt an alignment algorithms running in parallel with the proposed control laws to relax the assumption that all agents have been pre-aligned \cite{Oh2013OA,Montijano2016,aranda2016distributed}.
\subsection{Control without estimators}
\subsubsection{Proposed control law}
Select positive numbers $k_1$ and $k_2$ such that $k_2>\delta_2+k_1\delta_1$, the control law for follower $i \in \mc{F}$ is proposed as follows
\begin{align}
    \m{u}_i&=-k_1\m{v}_i-k_2\text{sign}\left(\m{s}_i\right),\label{Control 1}\\
    \m{s}_i&=\sum_{j\in \mc{N}_i} \m{P}_{\m{g}_{ij}^*}\ (\m{v}_i-\m{v}_j+k_1(\m{p}_i-\m{p}_j)).\label{SMC 1}
\end{align}
The controller $\m{u}_i$ is designed to guarantee the sliding variable $\m{s}_i\to \m{0}$ in finite time, thus the bearing constraints are also achieved in finite time. In \eqref{SMC 1}, the velocity term $\m{v}_i - \m{v}_j$ and the relative position $\m{p}_i-\m{p}_j$ are measured by the agent $i$. This implies that \eqref{SMC 1} is a decentralized control law.
\subsubsection{Stability analysis}
\begin{Lemma}\label{LemLem norm}
The following system
\begin{align}
    \dot{\m{x}}=-\m{A}\left(k\cdot\text{sign}\left(\m{x}\right)+ \m{1}_n\otimes \m{d}\right)
\end{align}
is finite time stable if $\m{A}$ is symmetric, positive-definite and the disturbance $\m{d}$ satisfies $\text{sup}_{t\ge 0}{\Vert \m{d}\left(t\right)\Vert}<k$.
\end{Lemma}
\begin{proof}
Taking the derivative of the Lyapunov function $V=0.5\m{x}^{\top}\m{A}^{-1}\m{x}$ yields  
\begin{align}
    \dot{V}=-k\Vert \m{x}\Vert_1-\m{x}^{\top}\left(\m{1}_n\otimes \m{d}\right).
\end{align}
The norm inequalities imply that
\begin{align}\label{Lemma norm}
    -\m{x}^\top\left(\m{1}_n\otimes \m{d}\right)\le \Vert \m{x}\Vert_1\Vert \m{d}\Vert_\infty\le \Vert \m{x}\Vert_1\Vert \m{d}\Vert.
\end{align}
Thus, $\xi = k-\text{sup}_{t\ge 0}{\Vert \m{d} \left(t\right)\Vert}>0$ satisfies $\dot{V}\le-\xi\Vert \m{x} \Vert_1$. The definition of  $V$ implies that
\begin{align}
    2V\le\lambda_{\max}\left(\m{A}^{-1}\right)\Vert \m{x}\Vert^2\le\lambda_{\max}\left(\m{A}^{-1}\right)\Vert \m{x}\Vert^2_1.
\end{align}
Combining \eqref{Lemma norm}, it holds $\dot{V} \le-\xi\sqrt{2\lambda_{\max}^{-1}\left(\m{A}^{-1}\right)V}$. Thus, based on Lemma~\ref{lem:finite-time} in the Appendix, we conclude that $V=0$ in finite time, and this implies that $\m{x}=\m{0}$ in finite time.
\end{proof}
\begin{Theorem}\label{Theorem 1}
Under control law \eqref{Control 1}, $\m{g}_{ij}\to \m{g}_{ij}^*$, $\m{p}_i \to \m{p}_i^*$, and $\m{v}_i\to \m{v}_c$ as $t\to\infty$. 
\end{Theorem}
\begin{proof}
Let $\m{\phi}_i=\sum_{j\in \mc{N}_i}{\m{P}_{\m{g}_{ij}^*}\left(\m{p}_{i}-\m{p}_{j}\right)}$, then $\m{\phi}_i=\m{0}_d$ if $i$ is a leader. Denoting $\m{\phi}=\text{vec}\left(\m{\phi}_1,\ldots,\m{\phi}_n\right)$ and $\m{\phi}_F=\text{vec}\left(\m{\phi}_{l+1},\ldots,\m{\phi}_n\right)$, it holds
$\m{\phi}=\m{ZBp}$, where $\m{Z}=\text{blkdiag}(\m{0}_{dl\times dl},\m{I}_{dl})$. Based on Lemma \ref{Lemma 1}, we have $\m{\phi}=\m{Z}\m{B}(\m{p}-\m{p}^*)$, or
\begin{align}
\left(\begin{matrix} \m{0}_{dl}\\\m{\phi}_F\\\end{matrix}\right)=\left(\begin{matrix}\m{0}_{dl\times dl}& \m{0}\\ \m{0}& \m{I}_{df}\\\end{matrix}\right)\left(\begin{matrix}\m{B}_{ll}&\m{B}_{lf}\\\m{B}_{fl}&\m{B}_{ff}\\\end{matrix}\right)\left(\begin{matrix} \m{0}_{dl}\\ \m{p}_f-\m{p}_f^*\\\end{matrix}\right)	
\end{align}
Thus, we obtain $\m{\phi}_F=\mcl{B}_{ff}\left(\m{p}_F-\m{p}_F^\ast\right)$. From the definition of $\m{\phi}_i$ and \eqref{SMC 1}, by denoting $\m{s}_F=\text{vec}(\m{s}_{l+1},\ldots,\m{s}_n)$, it holds $\m{s}_F={\dot{\m{\phi}}}_F+k_1\m{\phi}_F$. Now, we will prove that the sliding variable $\m{s}_F\to \m{0}$ as $t\to\infty$. Using  $\m{\phi}_F=\mcl{B}_{ff}\left(\m{p}_F-\m{p}_F^*\right)$, rewriting $\m{s}_F$ yields
\begin{align}
   \m{s}_F=\mcl{B}_{ff}\left(\m{v}_F-\m{v}_F^*+k_1\left(\m{p}_F-\m{p}_F^*\right)\right), 
\end{align}
where $\m{v}_F^*={\dot{\m{p}}}_F^*=\m{1}_f\otimes \m{v}_c$ and $\m{v}_F=\text{vec}(\m{v}_{l+1},\ldots,\m{v}_n)$. Taking the derivative of the sliding variable yields
\begin{align}
    {\dot{\m{s}}}_F=-\mcl{B}_{ff}\left(k_2\text{sign}\left(\m{s}_F\right)+\m{1}_f\otimes\m{\omega}_c\right),
\end{align}
where $\m{\omega}_c= \dot{\m{v}}_c + k_1\m{v}_c$ satisfies $\Vert \m{\omega}_c\Vert\le \delta_2+k_1\delta_1$. Based on Lemma \ref{LemLem norm}, we obtain $\m{s}_F\to \m{0}$; and thus, $\m{\phi}_F\to \m{0}$. Note that $\mcl{B}_{ff}$ is positive-definite, we conclude that $\m{p}_F \to \m{p}_F^*$ in finite time. Combining with $\m{s}_F \to \m{0}$, we have $\m{v}_F\to \m{1}_f \otimes \m{v}_c$, and this means $\m{v}_i \to \m{v}_c$ in finite time. The desired bearing constraints are attained directly from $\m{p}_F \to \m{p}_F^*$.
\end{proof}
\subsection{Control with estimators}
\subsubsection{Proposed control law}
Select positive numbers $k_q,~q=1,\ldots,6$ such that $k_5>\delta_2+k_4\delta_1$. Let ${\bar{\m{p}}}_i$, ${\hat{\m{p}}}_i$, $\m{p}_i$ be the estimate of the desired position, the estimate of current position of agent $i$. For each leader $i=1, \ldots, l$, we have ${\bar{\m{p}}}_i={\hat{\m{p}}}_i=\m{p}_i$ and ${\bar{\m{v}}}_i=\m{v}_c$. The PD-like control law for follower $i$ is written as 
\begin{align}\label{Control 2}
   \m{u}_i={\bar{\m{u}}}_i+k_1\left({\bar{\m{p}}}_i-{\hat{\m{p}}}_i\right)+k_2({\bar{\m{v}}}_i- \hat{\m{v}}_i), 
\end{align}
where ${\hat{\m{p}}}_i$ is estimated position, ${\bar{\m{p}}}_i$ and ${\bar{\m{v}}}_i$ denote the estimates of the desired position and the velocity of agent $i$. These variables are computed according to 
\begin{align}
    {\dot{\hat{\m{p}}}}_i&=\hat{\m{v}}_i \label{Consensus law1}\\
     {\dot{\hat{\m{v}}}}_i&= \m{u}_i + k_3\sum_{j\in\mc{N}_i}({\hat{\m{p}}}_j-{\hat{\m{p}}}_i+\m{p}_j-\m{p}_i) \nonumber \\& \qquad \qquad + k_6\sum_{j\in\mc{N}_i}({\hat{\m{v}}}_j-{\hat{\m{v}}}_i+\m{v}_j-\m{v}_i) , 
\label{Consensus law2}\\
    {\dot{\bar{\m{p}}}}_i&={\bar{\m{v}}}_i,\label{SMC 2a}\\
    {\dot{\bar{\m{v}}}}_i&={\bar{\m{u}}}_i=-{k_4\bar{\m{v}}}_i-k_5\text{sign}(\m{s}_i),\\
    \m{s}_i&=\sum_{j\in \mc{N}_i}{\m{P}_{\m{g}_{ij}^\ast}\left({\bar{\m{v}}}_i-{\bar{\m{v}}}_j+k_4\left({\bar{\m{p}}}_i-{\bar{\m{p}}}_j\right)\right)}.\label{SMC 2b}
\end{align}
The estimator \eqref{Consensus law1}--\eqref{Consensus law2} is a consensus law, which estimates the true position $\m{p}_i$ by the estimation ${\hat{\m{p}}}_i$ based on the relative position and relative velocity of its neighbors. Note that each agent can measure its velocity $\m{v}_i$ and receive the estimated position ${\hat{\m{p}}}_j$ of its neighbors. The sliding mode-based control law \eqref{SMC 2a}-\eqref{SMC 2b} computes the desired position of agent $i$ based on the desired bearing vectors and information obtained from its neighbors including ${\bar{\m{p}}}_j$ and ${\bar{\m{v}}}_i$. Next, we will prove that ${\hat{\m{p}}}_i\to \m{p}_i$, ${\bar{\m{p}}}_i\to \m{p}_i^*$, and then $\m{p}_i\to{\bar{\m{p}}}_i$. Compared with the first approach, this control law is more complicated, however, this can be extended to solve the bearing-based formation problem with the realistic issues such as input saturation, uncertainty parameters or optimization control, etc.
\subsubsection{Stability analysis}
\begin{Theorem}
Under control law \eqref{Control 2}, $\m{g}_{ij}\to \m{g}_{ij}^*$, $\m{p}_i\to \m{p}_i^*$, and $\m{v}_i\to \m{v}_c$ as $t\to\infty$. 
\end{Theorem}
\begin{proof}
Denote $\m{A}= \mcl{L}_{ff}\otimes \m{I}_d$, where $\mcl{L}_{ff}$ is defined as in \eqref{eq:Laplacian}. Note that the matrix $\mcl{L}_{ff}$ is symmetric positive-definite, and thus, $\m{A}$ is also symmetric positive-definite. From \eqref{SMC 2a}-\eqref{SMC 2b}, by following a similar process as in Theorem \ref{Theorem 1}, it is not hard to show that
\begin{align}\label{Like Theorem 1}
    {\bar{\m{p}}}_i\to \m{p}_i^*~\text{and}~ {\bar{\m{v}}}_i\to \m{v}_c
\end{align}
exponentially fast. Thus, $\bar{\m{u}}_i \to \dot{\m{v}}_c$ exponentially fast. 

We can rewrite the equations for $t\ge T$ by 
\begin{align*}
\dot{\m{p}}_F & = {\m{v}}_F \\
\dot{\m{v}}_F &= \bar{\m{u}}_F+ k_1 (\bar{\m{p}}_F - \hat{\m{p}}_F) + k_2 ( \bar{\m{v}}_F - \hat{\m{v}}_F) = \m{u}_F \\ 
\dot{\hat{\m{p}}}_F &= \hat{\m{v}}_F \\
\dot{\hat{\m{v}}}_F &=\m{u}_F + k_3  \m{A}(\hat{\m{p}}_F-\m{p}_F) + k_3  \m{A}(\hat{\m{v}}_F-\m{v}_F)
\end{align*}
where ${\m{v}}_F^* =  \m{1}_F \otimes \m{v}_c$
By letting $\m{\gamma} = {\m{p}}_F - \hat{\m{p}}_f$ and $\m{\delta} =  {\m{v}}_F - \hat{\m{v}}_f$, we can rewrite the system in the following form
\begin{align}
\dot{\m{\gamma}} &= \m{\delta} \\
\dot{\m{\delta}}  &= -k_3 \m{A}\m{\gamma} - k_6 \m{A} \m{\delta}.
\end{align}
Due to the fact that $\m{A}$ is positive definite and symmetric, it is not hard to see that the system $\m{\eta} = \m{M} \m{\eta}$, where $\m{\eta} = [{\m{\gamma}}^\top,~ \m{\delta}^\top]^\top$ and $\m{M} = \begin{bmatrix}
\m{0} & \m{I}_{df} \\  -k_3 \m{A} & - k_6 \m{A}
\end{bmatrix}
$. The eigenvalues of $\m{M}$ is determined from the characteristic equation $\text{det}(\lambda (\lambda\m{I}_{df} + k_6 \m{A}) + k_3 \m{A}) = 0$. Letting $0< \mu_1 \leq \ldots \leq \mu_{df}$ be eigenvalues of $\m{A}$, then the $2df$ eigenvalues of $\m{M}$ are roots of 
\begin{align*}
\lambda^2 + k_6 \mu_k \lambda + k_3 \mu_k = 0,~k=1, \ldots, df, 
\end{align*}
and it follows that Re$(\lambda_i)<0$, $\forall i = 1, 2, \ldots, 2df$. This implies $\m{\gamma}$ and $\m{\delta}$ converges to $\m{0}_{df}$ exponentially fast. Equivalently, $\hat{\m{p}}_F \to \m{p}_F$ and $\hat{\m{v}}_F \to \m{v}_F$ exponentially fast. 

Next, by writing $\m{u}_F = \m{1}_F \otimes \dot{\m{v}}_c + (\bar{\m{u}}_F - \m{1}_F \otimes \dot{\m{v}}_c) + k_1 (\m{p}_F^* - \m{p}_F + \m{p}_F - \hat{\m{p}}_F + \bar{\m{p}}_F - \m{p}_F^*) + k_2 ( {\m{v}}_F^* - \m{v}_F + \m{v}_F - \hat{\m{v}}_F + \bar{\m{v}}_F - \m{1}_F\otimes \m{v}_c)$, and letting $\m{x} = \m{p}_F - {\m{p}}_F^*$, $\m{y} = \m{v}_F - {\m{v}}_F^*$, we have
\begin{align}
\begin{bmatrix}
\dot{\m{x}}\\
\dot{\m{y}}
\end{bmatrix} = \begin{bmatrix}
\m{0}_{df} & \m{I}_{df}\\
- k_1 \m{I}_{df}  & - k_2 \m{I}_{df}
\end{bmatrix}
\begin{bmatrix}
{\m{x}}\\
{\m{y}}
\end{bmatrix} + \begin{bmatrix}
\m{0}_{df}\\
\m{z}(t)
\end{bmatrix}, \label{eq:xyz_system}
\end{align}
where $\m{z}(t) = (\bar{\m{u}}_F - \m{1}_F \otimes \dot{\m{v}}_c) + k_1(\bar{\m{p}}_F-\m{p}_F^*) + k_1\m{\gamma} + k_2(\bar{\m{v}}_F - \m{1}_F\otimes \m{v}_c) + k_2 \m{\delta}$. Due to the exponential stability of $\bar{\m{u}}_F - \m{1}_F \otimes \dot{\m{v}}_c$, $\bar{\m{p}}_F-\m{p}_F^*$, $\bar{\m{v}}_F - \m{1}_F\otimes \m{v}_c$, $\m{\gamma}$ and $\m{\delta}$, we have $\|\m{z}(t)\|\leq \|\bar{\m{u}}_F - \m{1}_F \otimes \dot{\m{v}}_c\| + k_1 \|\bar{\m{p}}_F-\m{p}_F^*\| + k_1\|\m{\gamma}\| + k_2\|\bar{\m{v}}_F - \m{1}_F\otimes \m{v}_c\| + k_2\| \m{\delta}\| \to 0$ exponentially fast. By considering $[\m{0}_{df}^\top,\m{z}^\top]$ as a vanishing input to the unforced system 
\begin{align}
\begin{bmatrix}
\dot{\m{x}}\\
\dot{\m{y}}
\end{bmatrix} = \begin{bmatrix}
\m{0}_{df} & \m{I}_{df}\\
- k_1 \m{I}_{df}  & - k_2 \m{I}_{df}
\end{bmatrix}
\begin{bmatrix}
{\m{x}}\\
{\m{y}}
\end{bmatrix}, \label{eq:xy_system}
\end{align}
and since the origin is an exponentially stable equilibrium of \eqref{eq:xy_system}, it follows from the input-to-state stability theory \cite{Khalil2002} that the system \eqref{eq:xyz_system} satisfies $\m{x} \to \m{0}_{df}$ and $\m{y} \to \m{0}_{df}$ asymptotically. This concludes the proof that the desired moving formation is asymptotically achieved.
%
\end{proof}

\begin{Remark} One can also design a finite-time estimator for \eqref{Consensus law1}--\eqref{Consensus law2} and finite-time control law instead of \eqref{Control 2} to make the overall desired formation  converges in finite time. 
\end{Remark}
\section{Simulation results}
In this section, we consider a five-agent formation in $\mathbb{R}^2$ with an interaction graph as shown in Fig.~\ref{fig: network}. The desired formation is specified by seven desired bearing vectors. The initial positions of the leaders satisfies their desired bearing constraint and they move with a same velocity. 

The simulation parameters are selected as follows
\begin{align}
\begin{cases}
\m{g}_{31}^*=\m{g}_{42}^*=\left[\begin{matrix}1\\0\\\end{matrix}\right],\ \m{g}_{32}^*=\m{g}_{54}^*=\frac{1}{\sqrt{2}}\left[\begin{matrix}1\\1\\\end{matrix}\right]\\
\m{g}_{34}^*=\left[\begin{matrix}0\\1\\\end{matrix}\right],\m{g}_{35}^*=\frac{1}{\sqrt{2}}\ \left[\begin{matrix}-1\\1\\\end{matrix}\right],\ \m{g}_{ij}^*=-\m{g}_{ji}^*\\
\m{p}_1\left(0\right)=\left[\begin{matrix}0\\0\\\end{matrix}\right],\m{p}_2\left(0\right)=\left[\begin{matrix}0\\2\\\end{matrix}\right]\ ,\ \m{v}_c=\left[\begin{matrix}1\\\sin{t}\\\end{matrix}\right].
\end{cases}
\end{align}
\begin{figure}
    \centering
    \includegraphics[width=.45\linewidth]{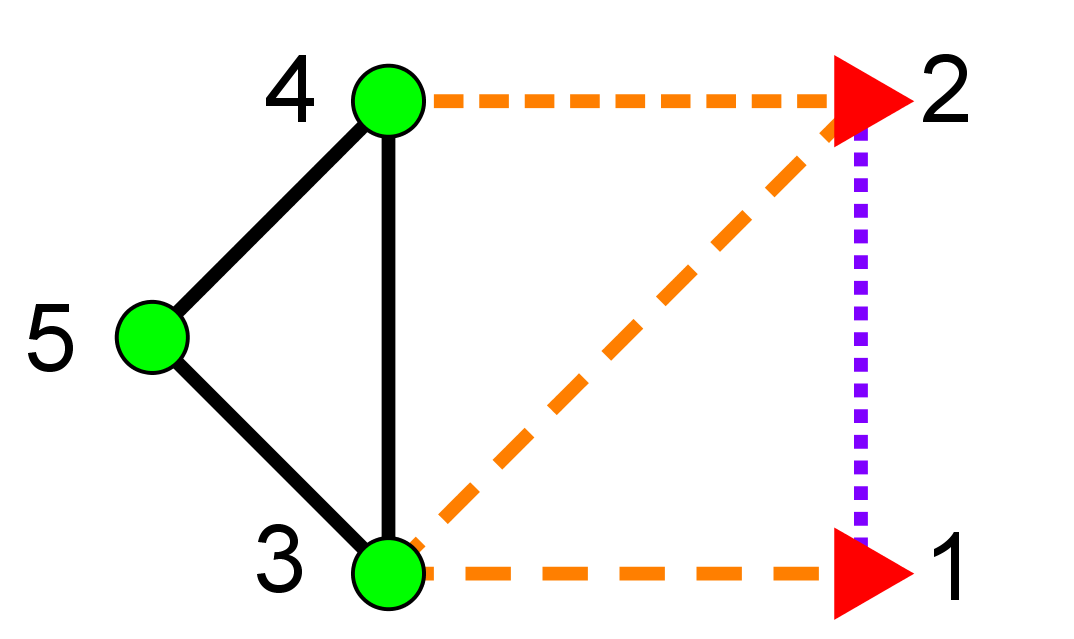}
    \caption{A five-agent network with 2 leaders (red triangles) and 3 followers (green circles).  }
    \label{fig: network}
\end{figure}

The initial positions of followers are randomly selected, and their initial velocities are chosen to be zero.  

In this section, two simulations are implemented for the five-agent formation. In the former (Sim. 1), we use the control law \eqref{Control 1} with parameters $k_1=0.5,~k_2=2$ while in the latter (Sim. 2), the control law \eqref{Control 2} is used with $k_1=k_2=k_3=1,~k_4=0.5,k_5=2$. The simulation results are shown in Fig. \ref{fig: tra1}-\ref{fig: v2}. Figures \ref{fig: tra1} and \ref{fig: tra2} depict the trajectories of agents where the initial and final positions are denoted by blue squares and red circles, respectively. As shown in Figs. \ref{fig: error1} and \ref{fig: error2}, the position errors $e_i=\Vert \m{p}_i-\m{p}_i^*\Vert^2$ and the bearing errors $e_{ij}=\Vert \m{g}_{ij}-\m{g}_{ij}^*\Vert^2$ tend to zero, which indicates that the bearing constraints are asymptotically satisfied, and the desired formation shape are obtained. The synchronization of agents’ velocities is illustrated in Figs. \ref{fig: v1} and \ref{fig: v2}, which guarantees that after obtaining the desired shape moving as a whole rigid formation. Thus, simulation results are consistent with the theoretical analysis. 

\section{Conclusions}
In this paper, we studied the bearing-based tracking control of leader-follower formations. Followers can measure the displacement, its velocity and exchange information with their neighbors. We introduce two sliding mode control-based control laws which solve the problem. Under the control laws, the desired formation shape is asymptotically achieved and the agents move with a common velocity. Our future work will focus on bearing-only control algorithms for the tracking control of leader-follower formations and consider further constraints on the agents' and network's model.
\begin{figure}[!t]
    \centering
    \includegraphics[width=\linewidth]{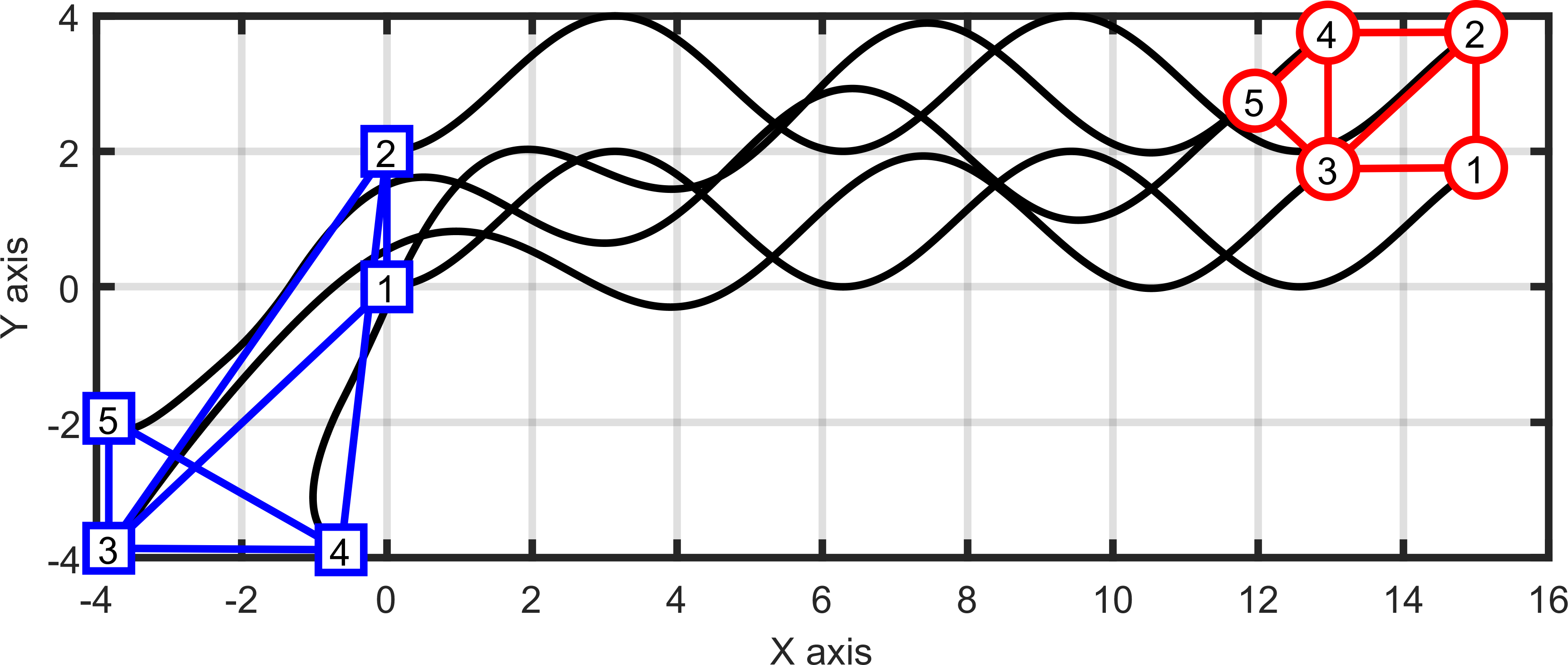}
    \caption{Sim 1: Trajectories of agents.}
    \label{fig: tra1}

    \includegraphics[width=\linewidth]{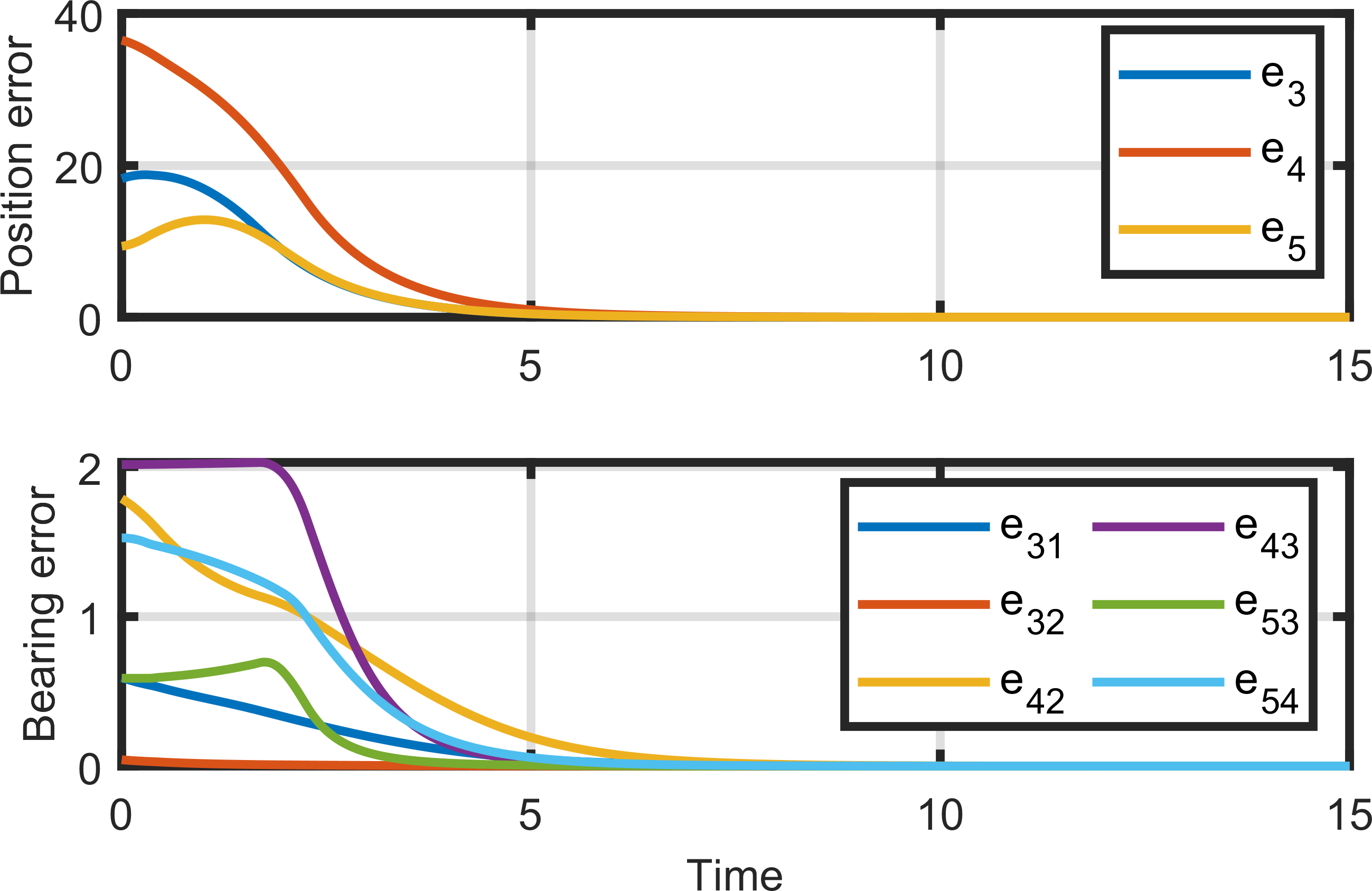}
    \caption{Sim 1: Position errors (above) and bearing errors (below) of agents.}
    \label{fig: error1}

    \includegraphics[width=\linewidth]{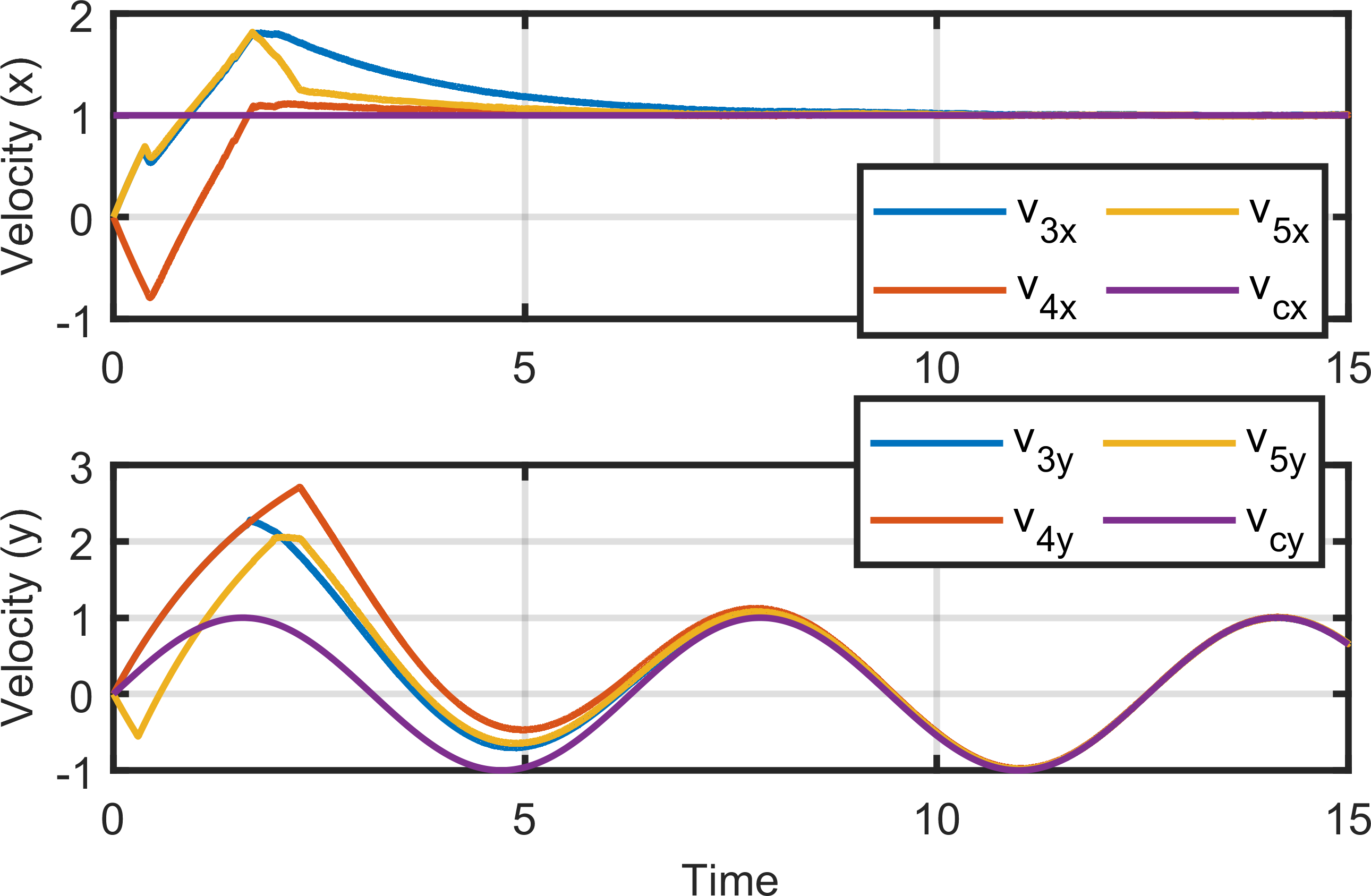}
    \caption{Sim 1: Agents’ velocity in Ox (above), Oy (below).}
    \label{fig: v1}
\end{figure}
\begin{figure}[th!]
    \centering
    \includegraphics[width=\linewidth]{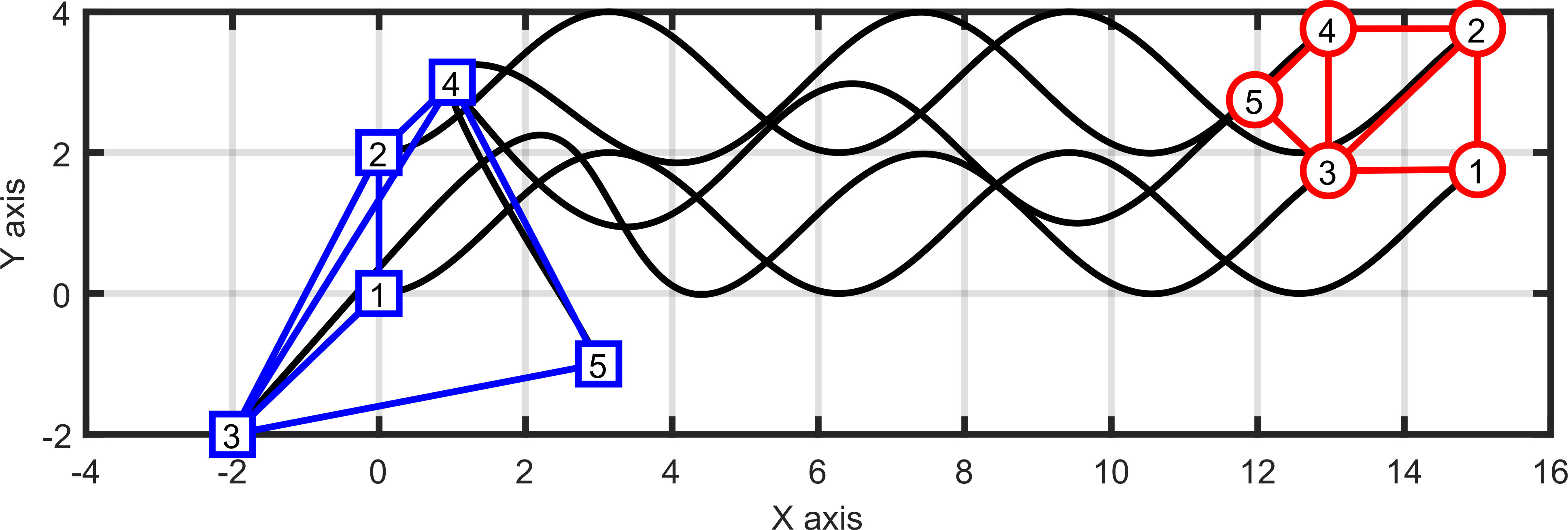}
    \caption{Sim 2: Trajectories of agents.}
    \label{fig: tra2}

    \includegraphics[width=\linewidth]{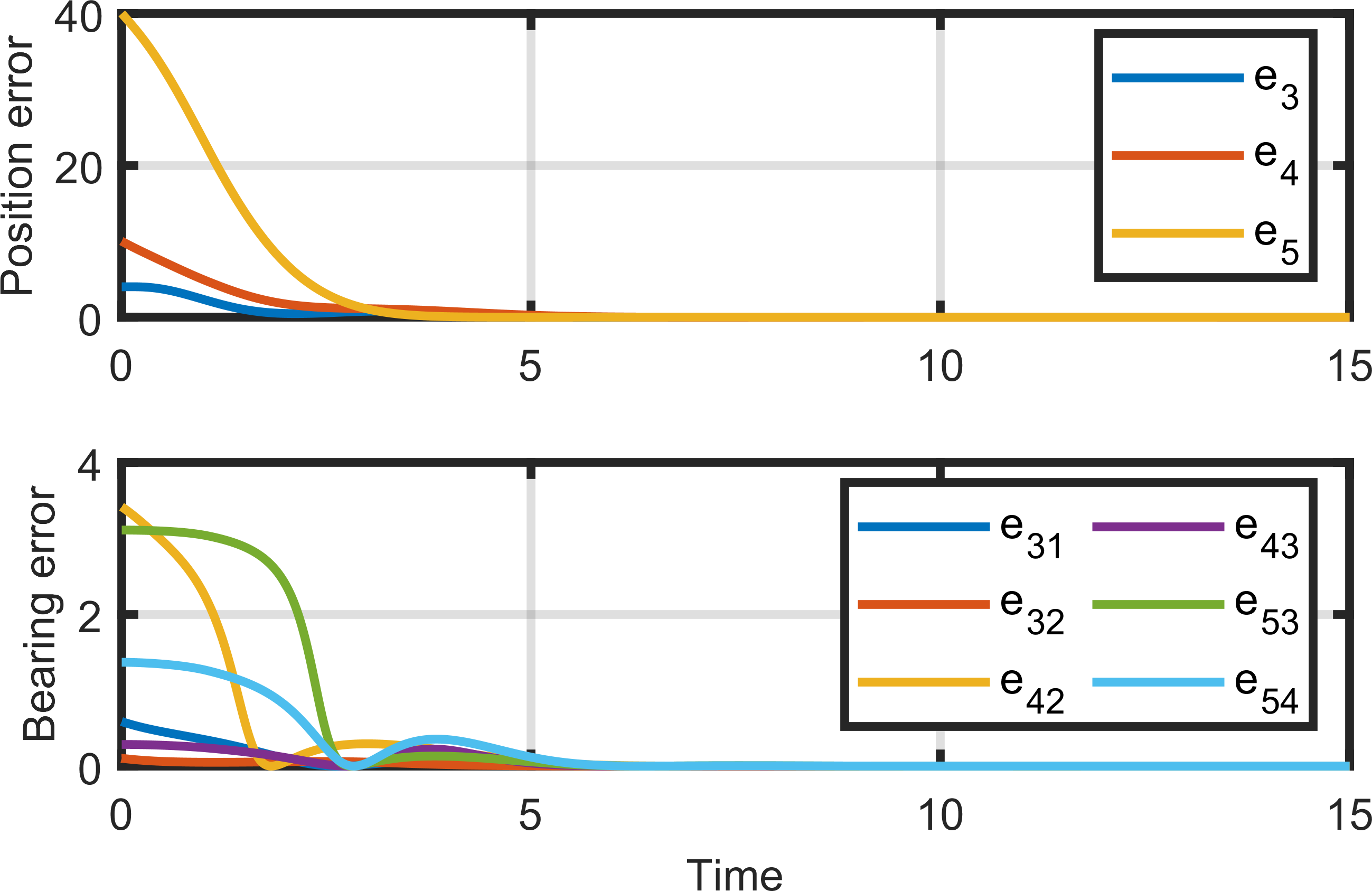}
    \caption{Sim 2: Position errors (above) and bearing errors (below) of agents.}
    \label{fig: error2}

    \includegraphics[width=1\linewidth]{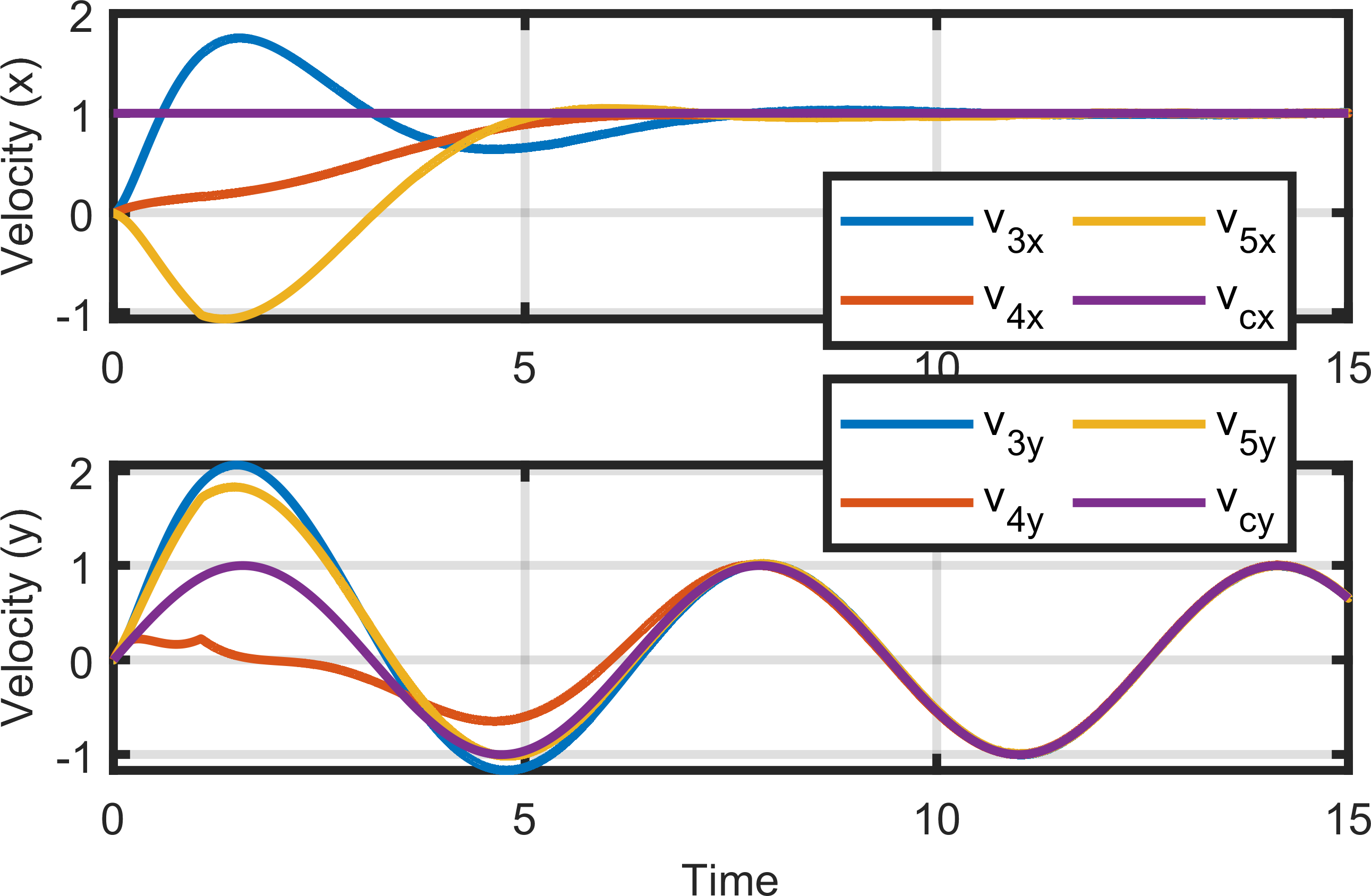}
    \caption{Sim 2: Agents’ velocity in Ox (above), Oy (below).}
    \label{fig: v2}
\end{figure}

\bibliographystyle{IEEEtran}
\bibliography{IEEEabrv,Ref}

\appendices
\section{Finite-time convergence theory}
\begin{Lemma} \label{lem:finite-time} \cite{Bhat2000} Let $\mc{D} \subseteq \mb{R}^d$ and $\m{x} \in \mc{D}$. Suppose there exists a continuous function $V(\m{x}): \mc{D} \to \mb{R}$ such that the following conditions hold
\begin{itemize}
\item[i)] $V(\m{x})$ is positive definite,
\item[ii)] If there exist $\kappa>0$, $\alpha \in (0,1)$, and an open neighborhood $\mc{U}_0 \in \mc{D}$ of the origin such that 
$$\dot{V}(\m{x}) + \kappa (V(\m{x}))^{\alpha } \leq 0, \forall \m{x} \in \mc{U}_0 \setminus \{\m{0} \},$$
\end{itemize}
then $V(\m{x})$ will reach zero in finite time with the settling time $T \leq V({0})^{1-\alpha}/(\kappa(1-\alpha))$.
\end{Lemma}
\end{document}